\documentclass{amsart}

\usepackage{amssymb,amsmath,euscript,enumitem}
\usepackage{lipsum}
\usepackage[margin=2in]{geometry}

\usepackage{}

\newtheorem{theorem}{Theorem}[section]
\newtheorem{lemma}[theorem]{Lemma}

\newtheorem{corollary}[theorem]{Corollary}

\theoremstyle{definition}
\newtheorem{definition}[theorem]{Definition}
\newtheorem{example}[theorem]{Example}

\theoremstyle{remark}
\newtheorem{remark}[theorem]{Remark}

\theoremstyle{algorithm}
\newtheorem{algorithm}[theorem]{Algorithm}
\def\Fq{{\mathbb F}_q}

\newcommand{\Fqm}{{\mathbb{F}_{q^m}}}
\newcommand{\Fqms}{\mathbb{F}_{q^m}^*}
\def\a{{\alpha}}

\newcommand{\Trqm}{\operatorname{Tr}}

\newcommand{\Ev}{\operatorname{Ev}}

\newcommand{\wt}{\operatorname{wt}}
\def\PP{{\mathbb P}}
\def\F{{\mathbb{F}}}

\def\X{{\mathbf{X}}}
\def\Y{{\mathbf{Y}}}
\def\Glm{{G_{\ell,m}}}
\def\G{{G_{2,m}}}
\def \Gm{{G_{2,m}}}

\def\Od{{\mathcal{O}_\delta}}

\def\Cg{{C(2, m)}}
\def\Cd {C^\delta(2,m)}

\def\a{{\alpha}}

\def\a {\alpha}
\def \Ilm {\mathbb{I}(\ell, m)}
\def \F {\mathbb{F}}

\def \Fq{\mathbb{F}_q}
\def \Fqm{\mathbb{F}_{q^m}}
\def \Clm {C(\ell, m)}

\begin{document}
	\title[Decoding Grassmann code $C(2, m)$]{The Orbit Structure of Grassmannian $G_{2, m}$ and Decoding  Grassmann code $C(2, m)$}
	\author{Fernando L. Pi\~nero }
	\address{Department of Mathematics,
		University of Puerto Rico at Ponce\newline \indent
		Ponce, P.R. 00716.}
	\email{fernando.pinero1@upr.edu}
	\author{Prasant Singh}
	\address{Department of Mathematics and Statistics,\newline \indent
		The Arctic University of Norway(UiT),\newline \indent
		Hansine Hansens veg 18, 9019 Tromsø, Norway.}
	\email{psinghprasant@gmail.com}
	\date{\today}
	\begin{abstract}

		In this article, we consider the decoding problem for Grassmann codes. We focus on the case of the Grassmannian of planes in an affine space. We look at the orbit structure of Grassmannian arising from the natural action of multiplicative group of a finite field extension. We project the corresponding Grassmann code onto these orbits to obtain a few subcodes of certain Reed-Solomon code. It is interesting to see that many of these projected codes contain an information set of the parent Grassmann code. By improving the decoding radius of  Peterson's decoder for the projected subcodes, we prove that one can correct up to $\lfloor (d-1)/2\rfloor$ errors for Grassmann code, where $d$ is the minimum distance of Grassmann code. 
		
	\end{abstract}
	\maketitle
	\section{Introduction}
	Let $q$ be a prime power, and $\Fq$  a finite field with $q$ elements. Let $\ell$ and $m$ be positive integers satisfying $
	\ell\le m$. The Grassmann variety $\Glm$ is a projective  algebraic variety over $\Fq$ whose points are $\ell$-dimensional subspaces of an $m$ dimensional vector space $V$ over $\Fq$. We may assume $V=\Fq^m$. To every such projective  variety, one can associate a  linear error-correcting code in a natural way, thinking the variety as a projective system \cite[Ch 1]{TVN}. The linear code associated to the Grassmanian $\Glm$, in this way, is known as the Grassmann code and  is denoted by $\Clm$. Ryan \cite{Ryan, Ryan2} initiated the study of Grassmann codes  over the binary field. Later, Nogin \cite{Nogin} continued the study of Grassmann codes over a general finite field and proved that the Grassmann code $\Clm$ is an $[n, k, d]_q$ code, where
	\begin{equation}
	\label{eq:parameters}
	n= {m \brack \ell}_q,\; k=\binom{m}{\ell},\text{ and } d=q^{\ell(m-\ell)}.
	\end{equation}
	Here ${m \brack \ell}_q$ denotes the Gaussian binomial coefficient given by
	\[
	{m\brack \ell}_q:=\dfrac{(q^m-1)(q^{m-1}-1)\cdots(q^{m-\ell+1}-1)}{(q^\ell-1)(q^{\ell-1}-1)\cdots(q-1)}.
	\]

	Mathematicians have been studying different aspects of Grassmann codes since they were discovered. For example, the weight spectrum of Grassmann codes $C(2, m)$, $C(3, 6)$ was computed by Nogin \cite{Nogin, Nogin1}.  Kaipa and Pillai \cite{KP} continued the study and computed the weight spectrum of the code $C(3, 7)$.  Some of the initial and terminal generalized Hamming weights of $\Clm$ are also known \cite{GL, GPP, Nogin}. Furthermore, the automorphism group of $\Clm$ is quite large and fully determined  \cite{GK}. Moreover, the dual codes of Grassmann codes have also been explored and some interesting geometric properties of the minimum weight of dual Grassmann codes has been discovered.  To be precise, it has been proven \cite[Thm. 24]{BP} that the support of the minimum weight codewords of $\Clm^\perp$ consists of three points from a line in the Grassmann variety $\Glm$. Conversely, any three points on a line in $\Glm$ are the support of a minimum weight codeword of $\Clm^\perp$. Other codes associated to different subvarieties of Grassmannian have also been of great interest. Schubert codes, linear codes associated to Schubert varieties in Grassmannian are one among these classes. The study of Schubert codes was initiated by Ghorpade-Lachaud \cite{GL} and a conjecture about the minimum distance of these codes was proposed. The conjecture is known as the MDC for Schubert codes. After several attempts and proving the MDC in many special cases \cite{HC, GV, GT}, this conjecture was settled in affirmative sense \cite{X, GS}. Like Grassmann codes, the weight spectrum of Schubert codes is also known in some cases \cite{PS}.

From the above discussion it is evident that several interesting problems related to Grassmann and related codes have been studied by different researchers over last three decades but the  decoding problem for these codes has not been explored in much detail.
So far, no effective decoding algorithm for Grassmann codes or Schubert codes is known. In a recent work, the second named author together with P. Beelen proposed a  decoding algorithm for Grassmann codes \cite{BS}. The authors exploited the geometric structures of Grassmann varieties to propose a decoder for Grassmann codes. In brief, they used geodesics in Grassmannians between distinct points of Grassmannian to construct certain orthogonal parity checks for Grassmann codes.  Then, a majority voting decoder is used to correct the errors. 
 But the proposed majority voting algorithm is not effective, since the proposed decoder can asymptotically correct approximately $d/2^{\ell+1}$ errors. In other words, even in the simplest cases, the decoder could not correct, asymptotically,  up to $\lfloor (d-1)/2\rfloor$ errors. Moreover, the second named author  extended  \cite{Prasant} the majority voting decoder for Grassmann codes $C(\ell, m)$  to Schubert codes corresponding to Schubert varieties in Grassmannian $\G$. Interestingly, in some cases, the proposed decoder for Schubert code is effective but in most of the cases it is not. Therefore, the problem of proposing an effective decoder for Grassmann and Schubert code is still open, even in the simplest case such as codes associated to Grassmannian $\G$ and Schubert varieties in $\G$.

	In this article, we study the decoding problem for Grassmann code $\Cg$. We consider the action of the cyclic group $\Fqms$ onto $\Gm$, thinking points of $\Gm$ as ordered pair of elements in $\Fqm$, and study the orbits of this action. We see that the projections of Grassmann code $C(2, m)$ onto these orbits are subcodes of certain Reed-Solomon codes. Moreover, most of these projections contain an information set of $C(2, m)$. We use such subcodes and Peterson's decoding algorithm to correct the errors. As a consequence, we are able to correct $\lfloor( d-1)/2\rfloor$ errors for the Grassmann code $C(2, m)$ where $d$ is the minimum distance of the code.

	\section{Preliminaries}
	In this section, we recall the  definition of Grassmann varieties and the construction of Grassmann codes. 
	Throughout the article, positive integers $\ell, m$ satisfying $\ell\le m$, and  finite field $\Fq$ with $q$ elements are fixed. The set $\Ilm$ is defined as 	\[
	\Ilm:=\{u=(u_1,u_2,\ldots, u_\ell)\in \mathbb{Z}^\ell: 1\le u_1<u_2<\cdots<u_\ell\le m\}
	\]
	and fix a linear order on $\Ilm$. Let $V$ be an $m$-dimensional vector space over $\Fq$. The Grassmannian $G_\ell(V)$ of all $\ell$-planes of vector space $V$ is defined by
	\[
	G_\ell(V):=\{P: P\subset V\text{ is a linear subspace and }\dim P=\ell\}.
	\]
	The Grassmannian $G_\ell(V)$ can be embedded into a projective space via the Pl\"ucker map. More precisely, let $\mathcal{B}$ be an ordered basis of $V$. For $P\in G_\ell(V)$, let $M_P$ be an $\ell\times m$ matrix whose rows are coordinates of some basis of $P$ with respect to  $\mathcal{B}$. The Pl\"ucker map $\mathrm{Pl}$ is defined by:
	\begin{equation}
	\label{eq: plucker}
	\mathrm{Pl}:G_\ell(V)\to \mathbb{\PP}^{{m\choose\ell}-1} \text{ defined by }P\mapsto [(M_P)_u]_{u\in\Ilm}
	\end{equation}
	where $(M_P)_u$ denotes the $\ell\times\ell$ minor of $M_P$ corresponding to columns of $M_P$ indexed by tuple $u$. The image of the Pl\"ucker map is given by the zero set of a bunch of quadratic polynomials known as Pl\"ucker polynomials and hence defined a projective algebraic variety, known as Grassmann variety. For a detailed study on Grassmann varieties we refer to \cite{KL, Manivel}. It is known that if $V$ and $V'$ are two vector spaces of dimension $m$ over $\Fq$, then there exists an automorphism of $\mathbb{\PP}^{{m\choose\ell}-1}$ mapping  $G_\ell(V)$ to $G_\ell(V')$. Therefore, for the rest of the article, we denote by $\Glm$, the Grassmannian of all $\ell$-planes of $\Fq^m$. Using the Pl\"ucker map, we may think $\Glm$ as a subset of the projective space $\mathbb{\PP}^{{m\choose\ell}-1}$ over finite field. To every projective algebraic variety defined over a finite field, one can associate a linear code using the language of projective systems  \cite[Ch.1]{TVN}. More precisely, each nondegenerate subset $\mathcal P$ of a projective space $\PP^{k-1}$ over $\Fq$ corresponds to a unique linear code. Further, the minimum distance and the generalized Hamming weights of the corresponding code can be studied from hyperplane sections of $\mathcal P$ with linear subspaces of $\PP^{k-1}$. Therefore, Grassmannian $\Glm$ corresponds to a linear code. The code associated to $\Glm$ in this way is known as the Grassmann code and is denoted by $\Clm$.  To go into more details, we would like to recall the construction of Grassmann code. 
	
	Let $\mathbf{\underline{X}}=(X_{ij})$ be an $\ell\times m$ matrix of $\ell m$ indeterminates $X_{ij}$ over $\Fq$. For every $u\in\Ilm$, let $\mathbf{X}_u$ denote the $\ell\times\ell$ minor of $\mathbf{X}$ corresponding to columns labeled by $u$. Let $\Fq[\X_u]_{u\in\Ilm}$ be the $\Fq$ linear space spanned by all minors $\X_u$. For each $P\in\Glm$, let $M_P$ be a matrix corresponding to point $P$ as in equation \eqref{eq: plucker}, and let $\{M_{P_1},\dots,M_{P_n}\}\subset \Fq^{\ell\times m}$ be  a set of such matrices corresponding to each point $P_i\in\Glm$ in some fixed order, where $n=|\Glm|$. Consider the evaluation map
	\begin{equation*}
	\label{eq:constGrass}
	\Ev: \Fq[\X_u]_{u\in\Ilm}\to\Fq^n \text{ defined by }f(\X_u)\to(f(M_{P_1}),\dots,f(M_{P_n})).
	\end{equation*}
	The image of the evaluation map $\Ev$ is known as the Grassmann code and is denoted by $\Clm$. The Grassmann code $\Clm$ is an $[n,k,d]_q$ linear code where $n, k$ and $d$ are given by equation~\eqref{eq:parameters}. Clearly, the codewords of  Grassmann code $\Clm$ are indexed by points of $\Glm$. Therefore, we may use points $P_1,\ldots, P_n\in\Glm$ as an indexing set for the coordinates of codewords in $\Clm$.

	In this article, we only study the Grassmann code $C(2, m)$. We write a $2\times m$ generic matrix as
	$$
	\mathbf{\underline{X}}=
	\begin{bmatrix}
	X_1 & X_2 & \cdots & X_m \\
	Y_1 & Y_2 & \cdots & Y_m 
	\end{bmatrix}
	$$
	and write the first row of the indeterminate matrix $\underline\X$ as $\X$  and the second row of $\underline\X$ as $\Y$. In the next section, we will study the orbit structure of Grassmannian $G_{2, m}$ under the natural action induced by the cyclic group $\Fqms$ but before we get into the orbit structure, we recall  the definition of the \emph{ trace}  function of field extensions.
	
	Let $\Fqm$ be the field extension of $\Fq$ of degree $m$.
	\begin{definition}
		The \emph{trace function} of $\Fqm$ over $\Fq$ is defined and denoted by
		$$
		\Trqm_{\Fqm/\Fq}(x) = x + x^q+x^{q^2} + \cdots + x^{q^{m-1}}.
		$$
		If the fields $\Fqm$ and $\Fq$ are clear from the context, we drop the index and denote the trace map by $\Trqm$.
	\end{definition}
	Note that $\Fqm$ is an $m$ dimensional vector space over $\Fq$ and $\Trqm$ is an $\Fq$ linear map.  Trace functions will play a key role in our decoder.
	
	\section{The Orbit Structure of Grassmannian  $G_{2, m}$}
	
	In this section, we study the natural action of the cyclic group of $\Fqms$ on Grassmannian $G_{2,m}$. Our goal is to understand the orbits of $\G$ under this action and the behavior of the projection of the code $C(2,m)$ onto these orbits.  Before going into any further details, we shall fix some notations that will be used throughout the article. As earlier, let $\Fqm$ be the field extension of $\Fq$ of degree $m$. We know that $\Fqm$ and $\Fq^m$ are isomorphic as  an $\Fq$ vector space. We fix an isomorphism between $\Fq^m$ and $\Fqm$. Using this isomorphism, we may think of $\Gm$ as a subset of $\mathbb{F}_{q^m}^*\times\Fqms$ consisting of tuples $\langle\a,\beta\rangle_{\Fq}$, where $\a,\;\beta\in\Fqms$ span a two dimensional subspace over $\Fq$. 
	For $\a\in\Fqm$, we denote by $(\a_1,\ldots,\a_m)\in\Fq^m$, the coordinates of $\a$ and vice-versa. Furthermore, we treat the subspace $\langle\a,\beta\rangle_{\Fq}$ as the point of $\Gm$ spanned by coordinates $(\a_1,\ldots,\a_m)$ and $(\beta_1,\ldots,\beta_m)$. 
Recall the following trivial lemma:
	
	\begin{lemma}
		\label{lemma:trace}
		The map $T: \Fqm\times \Fqm\to \Fq$ defined by $T(\a, X)= \Trqm(\a X)$ is a non-degenerate, $\Fq$-bilinear map. In particular, 
		for every $\alpha \in \Fqms$, the map  $X \mapsto Tr_{\Fqm/\Fq}(\a X)$ is a nonzero $\Fq$-\emph{linear functional} of $V = \Fqm$.
	\end{lemma}
	Note that, the trace map is $\Fq$-linear in both $X$ and $\a$. Therefore, for every $X\in\Fqm$ and  $X_i\in\Fq$, there exist some $\a\in\Fqm$ such that $Tr_{\Fqm/\Fq}(\a X)=X_i$. Furthermore, if $X_i\in\Fq$ is a coordinate of $X\in\Fqm$ (via isomorphism treating it in $\Fq^m$), then there exist some $\a\in\Fqms$ such that $Tr_{\Fqm/\Fq}(\a X)=X_i$. This plays a very important role. The next lemma is an immediate consequence of Lemma \ref{lemma:trace}.  
	
	\begin{lemma}
		\label{lemma:determinant}
		Let $\alpha, \beta \in \Fqms$ and let $W =  \Fqm \times \Fqm$ be the $\Fq$ vector space. Then functions of the form  
		$$
		f_{\alpha, \beta}(X,Y) =  \Trqm(\alpha X) \Trqm(\beta Y) - \Trqm(\alpha Y) \Trqm(\beta X) 
		$$ 
		are determinantal functions on $W$ as a vector space over $\Fq$.
	\end{lemma}
	What we mean is that the function $f_{\alpha, \beta}(X,Y)$ is an alternating bilinear  map on $W$. Therefore, $2\times 2$ minors $X_iY_j-X_jY_i$ of the $2\times m$ matrix $\underline\X$, can be written as $\Fq$-linear combination of functions $	f_{\alpha, \beta}(X,Y)$ for $\alpha, \beta \in \Fqms$. In other words, one can think of Grassmann code as evaluation functions $\sum_{a\in\Fq}\sum_{\a,\beta}af_{\alpha, \beta}(X,Y)$ on Grassmannian $\Gm$ as a subset of $\Fqms\times\Fqms$. We will return to these functions and their evaluations in the next section.
	%
	
	Now we are ready to look at the natural group action of $\Fqms$ on Grassmannian $\Gm$. First, we shall define the action of $\Fqms$ on  $\Gm$.

	\begin{definition}
			\label{def:grpaction}
		Let $\gamma$ be a generator of the cyclic group $\Fqms$. The action of $\Fqms$ on Grassmannian $\Gm$ is defined by
	\begin{equation}
	\label{eq:grpaction}
 \gamma^i\cdot\langle \alpha, \beta \rangle_{\Fq}= \langle \gamma^i\alpha, \gamma^i\beta \rangle_{\Fq}.
	\end{equation}
\end{definition}
 For the rest of the article, we fix this action of $\Fqms$ on $\Gm$ and whenever we talk about group action on $\Gm$ or orbits of $\Gm$, we always mean the action defined in Definition \ref{def:grpaction}.	Let $O_1,\ldots, O_r$ be all the orbits of $\Gm$. Therefore, if $\langle \alpha, \beta \rangle_{\Fq}\in O_j$ is an arbitrary element of orbit $O_j$ then $O_j=\{\langle \gamma^i\alpha, \gamma^i\beta \rangle_{\Fq}: 1\le i\le q^m-1\}$. Since $\gamma$ generates $\Fqms$ and for  any $\langle \alpha, \beta \rangle_{\Fq}\in \Gm$, we have $\a\in\Fqms$, therefore, we may assume that each orbit has an element of the form $\langle 1, \delta \rangle_{\Fq}$. We denote the orbit containing $\langle 1, \delta \rangle_{\Fq}$ as $O_\delta$ and we call the element $\langle 1, \delta \rangle_{\Fq}$ an orbit representative of $O_\delta$.
	
	\begin{example}
		\label{example:counting}
		Assume  $m=4$ and $q=2$. Let $\gamma$ be such that $\gamma^4 = \gamma+1$. In this case $\mathbb{F}_{16} = \mathbb{F}_2[\gamma]$ and $\langle \gamma \rangle = \mathbb{F}_{16}^*$. Note that if $\langle 1, \a \rangle_{\Fq}$ is an orbit representative of the orbit $O_\a$, then  $\langle 1, \a \rangle_{\Fq}=\{0, 1, \a, 1+\a\}$. By removing $0$ we may just write this set as $\{ 1,\a, \a +1\}$. Further, as $\gamma$ is a generator of the field $\mathbb F_{16}$, we have $\a=\gamma^i$ and $1+\a=\gamma^j$ for some $i$ and $j$. As the subspace  $\langle 1, \alpha \rangle_{\Fq}$ is of dimension $2$, we get $\alpha \neq 0, 1$. This leaves $16-2 = 14$ possibilities for $\a$. Further, each of those $\alpha$ and $\alpha+1$ generates the same space. Therefore, there are $7$ different spaces of the form $\langle 1, \alpha \rangle$. These are $\{ 1, \gamma, \gamma^{4} \}$, $\{ 1, \gamma^{2}, \gamma^{8} \}$, $\{ 1, \gamma^{14}, \gamma^{3} \}$, $\{ 1, \gamma^{5}, \gamma^{10} \}$, $\{ 1, \gamma^{13}, \gamma^{6} \}$, $\{ 1, \gamma^{7}, \gamma^{9} \}$, and $\{ 1, \gamma^{11}, \gamma^{12} \}$. Since the action of $\gamma$ maps $\{ 1, \gamma^i, \gamma^j  \}$ to $\{ \gamma, \gamma^{i+1}, \gamma^{j+1}\}$.
		A direct computation shows that there are three orbits, namely: The orbits  $O_{\gamma}$, $O_{\gamma^2}$ and $O_{\gamma^5}$. The orbit $O_{\gamma}$ has 15 elements and contains orbit representatives $\{ 1,\gamma, \gamma^4 \},\;\{ 1, \gamma^3, \gamma^{14} \}$, and $\{ 1, \gamma^{11}, \gamma^{12} \}$.  The orbit $O_{\gamma^2}$ also has  15 elements and it  contains orbit representatives  $\{1, \gamma^2, \gamma^8\}\; \{ 1, \gamma^{6}, \gamma^{13} \}$, and $\{ 1, \gamma^{7}, \gamma^{9} \}$. Finally the orbit $O_{\gamma^5}$  has 5 elements and it has only one orbit representative, namely $\{1, \gamma^5, \gamma^{10}\}$. Note that this gives in total 35 spaces, i.e. full Grassmannian $G_{2,4}$.

		The Grassmannian $G_{2,5}$ is represented as elements of $\mathbb{F}_{32}$. Here, the field $\mathbb{F}_{32}$ has $31$ nonzero elements. If $\langle 1, \a \rangle_{\Fq}$ is an orbit representative of the orbit $O_\a$, then there are 30 choices for $\a$. Further, since $\a$ and $1+\a$ generate the same space, we have only 15 choices of the subspace  $\langle 1, \a \rangle_{\Fq}$. Furthermore, if $\gamma$ is a generator of $\mathbb{F}_{32}^*$, then the action of $\gamma$ does not fix any elements of $G_{2,5}$.  Thus, all orbits will have size $31$ and hence there are $15$ orbits of size $31$. Likewise, the Grassmannian $G_{2,7}$ represented as elements of $\mathbb{F}_{128}$  has $63$ orbits of size $127$.
	\end{example}
	In the next two lemmas we will understand why in the case of $m=5$ and $m=7$, the orbits structure of $\Gm$ are quite uniform. 
	
	\begin{lemma}\label{lemma:cyclefix}
		Let $P = \langle 1, \alpha \rangle_{\Fq} $ be a two dimensional  $\F_q$-linear subspace of $\mathbb{F}_{q^m}$. Suppose that $\beta \in \mathbb{F}_{q^m}\setminus  \F_q$. Then $\beta P = P$ if and only if $\alpha \in \F_{q^2}^*$, i.e. $P=\F_{q^2}$.
		
	\end{lemma}
	\begin{proof}
		Let $P=\langle 1, \alpha \rangle_{\Fq}= \{ a + b\alpha  \ | \ a,b \in \F_q \}$ be as in the hypothesis and let $\beta \in \mathbb{F}_{q^m}\setminus\Fq$  such that $\beta P = P$.  Then, as $\beta \in \beta P$ and $\beta P = P$, we have  $\beta = a + b \alpha$ for some $a,b \in \F_q$. Since $\beta \not \in \Fq$ we have $b \neq 0$. Also, as $1 \in \beta V$ there exist $c, d \in\F_q$ such that $1 = (c+d\alpha)\beta$.  Putting the value of $\beta$, we get $1 = (c+d\alpha)(a+b\alpha)$. Therefore, $1 = bd\alpha^2 + (ad+bc)\alpha+cd$ and hence $\alpha$ satisfies a polynomial equation over $\F_q$ of degree $2$.  It follows that $\alpha \in \F_{q^2}$. For the reverse implication, note that if $\alpha \in \F_{q^2} \setminus \F_q$, then $P = \F_{q^2}$ and in this case, clearly $\alpha P = P$.

	\end{proof}
	In the next lemma, we consider $\G$ when $m$ is odd and count the number of orbits in $\G$ under the action defined in equation \eqref{eq:grpaction} and compute the size of each orbit.
	\begin{lemma}\label{lemma:oddcycle}
		If $m$ is an odd integer, then there are $\frac{q^{m-1}-1}{q^2-1}$ orbits and   the size of each orbit is $\frac{q^m-1}{q-1}$.
	\end{lemma}
	\begin{proof}
		
		The proof is a simple consequence of the orbit-stabilizer theorem and Lemma \ref{lemma:cyclefix}. Since $m$ is odd, there does not exist any $\a\in\Fqm\setminus \Fq$ such that $\a\in\mathbb{F}_{q^2}$. Therefore, from lemma \ref{lemma:cyclefix}, we know that for any $P\in\Gm$, the stabilizer of $P$ has size $(q-1)$, namely elements of $\Fq^*$ . Now from the orbit-stabilizer theorem, we get that the orbit of each of $P$ is of size $\frac{q^m-1}{q-1}$. Further, as $\Gm$ is the disjoint union of orbits of $P\in\Gm$, the number of orbits is $|\Gm|/|O_\delta|$, where $O_\delta$ is an arbitrary orbit. As a result, we get the total number of orbits in $\Gm$ are $\frac{q^{m-1}-1}{q^2-1}$.
	\end{proof}	
	This lemma justifies the nature of orbits of $\Gm$ in  cases $m=5,\;7$ that we discussed in the Example \ref{example:counting}. The next lemma counts the size of orbits and total number of orbits in $\Gm$ when $m$ is even.
	
	\begin{lemma}
		\label{lemma:evencycle}
		If $m$ is even, then there are  $q\frac{q^{m-2}-1}{q^2-1}$ orbits of size $\frac{q^m-1}{q-1}$ and exactly one orbit of size $\frac{q^m-1}{q^2-1}$.
	\end{lemma}
	
	\begin{proof}
		Let $P\in\Gm$ be an arbitrary element of the Grassmannian. If $P=\mathbb{F}_{q^2}$ then we have $\beta P=P$ if and only if $\beta\in\mathbb{F}_{q^2}^*$. In other words, the stabilizer of $P$ in this case is of size $q^2-1$ and hence from orbit-stabilizer theorem we get that the orbit of $P$ in this case is of size $\frac{q^m-1}{q^2-1}$. On the other hand, if $P\neq \mathbb{F}_{q^2}$, then from lemma \ref{lemma:cyclefix} we know that in this case we will have $\beta P=P$ if and only if $\beta\in\mathbb{F}_{q}^*$. In other words, in this case the stabilizer of $P$ is of size $q-1$ and hence the orbits of $P$ is of size $\frac{q^m-1}{q-1}$. Now as the cardinality of $\Gm$ is $\frac{(q^m-1)(q^{m-1}-1)}{(q^2-1)(q-1)}$ and if there are $r$ orbits of size $\frac{q^m-1}{q-1}$, then we have 
		\[
		\frac{(q^m-1)(q^{m-1}-1)}{(q^2-1)(q-1)}= r\frac{q^m-1}{q-1} \; +\; \frac{q^m-1}{q^2-1}.
		\]
		Solving for $r$ gives there are $q\frac{q^{m-2}-1}{q^2-1}$ orbits of size $\frac{q^m-1}{q-1}$.
	\end{proof}	
	\section{Evaluation of the determinant functions on Orbits of $\Gm$}
	
	Our bound and decoding algorithm of Grassmann code $C(2, m)$ hinges on the fact that the Grassmann code $C(2, m)$ is a subcode of quadratic forms in $2m$ variables, namely one  variable for each entry of the $2 \times m$ generic matrix. But in this section we will think this code in a slightly different way. We know that reordering the points of Grassmannian only gives an equivalent code. Therefore, we first fix an order on orbits and then order points in each orbit. We may think $C(2, m)$ as evaluation of determinantal functions on representatives of points in each orbit in these  fixed orders.  Recall that the determinantal functions can be written as a $\Fq$-linear combination of functions $f_{\a,\beta}(X, Y)$ where $\a,\;\beta\in\Fqm$. Also that, the orbit of $\langle 1, \delta\rangle_{\Fq}$ is the set $\Od = \{ \langle\gamma^i, \gamma^i\delta\rangle_{\Fq}\ | \ 0\le i\le q^m-1\}$ where $\gamma\in\Fqms$ is a generator of the multiplicative cyclic group $\Fqms$. Note that, when determinantal functions $f_{\a,\beta}(X, Y)$ are evaluated on an arbitrary point $\langle\gamma^i, \gamma^i\delta\rangle_{\Fq}$  of the orbit $\Od$, it gives
	$$
	f_{\a,\beta}(\gamma^i, \gamma^i\delta)= \Trqm(\alpha \gamma^i ) \Trqm(\beta\gamma^i\delta) - \Trqm(\alpha \gamma^i\delta) \Trqm(\beta \gamma^i).
	$$
	Since $\a, \beta$ and $\delta$ are fixed, we may think $f_{\a,\beta}(X,Y)$ as polynomial in one variable when it is evaluated on orbit $\Od$ in some fixed order. Hence, we consider polynomials 
	$$
	f_{\alpha, \beta}(T, \delta T) =   \Trqm(\alpha T) \Trqm(\beta \delta T) - \Trqm(\alpha \delta T) \Trqm(\beta T)
	$$
	where $\a,\;\beta\in\Fqm$. It implies that the evaluation of determinantal functions $f_{\a,\beta}(X, Y)$ on orbits $\Od$ is also given by the evaluation of polynomials $f_{\alpha, \beta}(T, \delta T)$ on  ``certain elements" of $\Fqms$. This evaluation also gives a linear code which we denote by $C^\delta(2,m)$. 
	This motivates the following definition.

	 \begin{definition}
	 	\label{Def: orbitcode}
	 	Let $\Od=\{P_1,\dots, P_N\}$ be an orbit in $\Gm$ with an orbit representative $\langle 1, \delta\rangle$ and let $P_i=\langle \gamma^{j_i},\gamma^{j_i}\delta\rangle_{\Fq}$.  Let $L_\delta(\Fq)$ be the $\Fq$- space spanned by the set $\{ f_{\alpha, \beta}(T, \delta T): \a,\;\beta\in\Fqm\}$, i.e., 
	 \begin{equation}
	 \label{def: ldfq}
	 L_\delta(\Fq)=\langle\{ f_{\alpha, \beta}(T, \delta T): \a,\;\beta\in\Fqm\}\rangle_{\Fq}.
	 \end{equation}
	 	Consider the evaluation map
	 \begin{eqnarray}
	 \label{eq: CodeCd}
	 	\Ev: 	L_\delta(\Fq)&\to& \Fq^N  \\
	 F(T)=\sum_{a, \a,\beta} af_{\alpha, \beta}(T, \delta T)&\mapsto& (F(\gamma^{j_1}),\ldots, F(\gamma^{j_N})). \nonumber
	 \end{eqnarray}
 The image of the evaluation map $\Ev$ is a code and we denote this code by $C^\delta(2,m)$.
	 \end{definition}
	
	We know that, all orbits (except possibly one)  in $\Gm$ are of size $\frac{q^m-1}{q-1}$. Therefore, almost all codes $C^\delta(2,m)$ are of length $N= \frac{q^m-1}{q-1}$.

%
	\begin{remark}
		\label{rmk:coeeficients}
		The determinatal function $ f_{\alpha, \beta}(T, \delta T)$ depends on $\alpha, \beta\in\Fqm$, therefore  when we think a function in $L_\delta(\Fq)$ as a polynomial in $T$, the coefficients are in the field $\Fqm$. But since determinantal functions $ f_{\alpha, \beta}(T, \delta T)$ are defined by Trace function, the evaluation of $ f_{\alpha, \beta}(T, \delta T)$ on points $P_i$ are in the field $\Fq$ and hence the code $\Cd$ is a code over the field $\Fq$.
	\end{remark} 

	\begin{lemma}
		The code $C^\delta(2,m)$ is a projection of the code $C(2,m)$ onto the coordinates in the orbit $\Od$.
	\end{lemma}
	\begin{proof}
		This follows from the fact that  polynomials in  $L_\delta(\Fq)$ give all determinantal functions and hence evaluation functions for Grassmann code, and that the Grassmannian $\Gm$ is the disjoint union of the orbits $\Od$.
	\end{proof}	
	\begin{lemma}
		\label{lemma:deg}
		If $\a, \beta\in\Fqm $ is such that $f_{\alpha, \beta}(T, \delta T)$ is a non-zero polynomial. Then  $q+1\le \deg(f_{\alpha, \beta}(T, \delta T))\leq q^{m-1}+q^{m-2}$. 
	\end{lemma}
	\begin{proof}
		We simply expand the determinantal function $f_{\alpha, \beta}(T, \delta T)$ using the trace function  $\Trqm(T) = \sum\limits_{i=0}^{m-1} T^{q^i}$. Note that, 
		
		\begin{eqnarray*}
			f_{\alpha, \beta}(T, \delta T)&=&\Trqm(\alpha T) \Trqm(\beta \delta T) - \Trqm(\alpha \delta T) \Trqm(\beta T)\\
			&=&  \left(\sum\limits_{i=0}^{m-1} (\alpha T)^{q^i}\right) \left(\sum\limits_{j=0}^{m-1} (\beta \delta T)^{q^j}\right) - \left(\sum\limits_{i=0}^{m-1} (\beta T)^{q^i}\right) \left(\sum\limits_{j=0}^{m-1}(\alpha \delta T)^{q^j}\right) \\
			&=& \left(\sum\limits_{i=0}^{m-1}\sum\limits_{j=0}^{m-1} (\alpha T)^{q^i}(\beta \delta T)^{q^j}\right) - \left(\sum\limits_{i=0}^{m-1}\sum\limits_{j=0}^{m-1} (\beta T)^{q^i} (\alpha \delta T)^{q^j}\right) \\
			&=&  \left(\sum\limits_{i=0}^{m-1}\sum\limits_{j=0}^{m-1} (\alpha T)^{q^i}(\beta \delta T)^{q^j}  - (\beta T)^{q^i} (\alpha \delta T)^{q^j}\right)\\
			&=& \left(\sum\limits_{i=0}^{m-1}\sum\limits_{j=0}^{m-1} (\alpha^{q^i} \beta^{q^j} - \beta^{q^i}\alpha^{q^j}) \delta^{q^j} T^{q^i+q^j}\right)\\
			&=&\left(\sum\limits_{i=0}^{m-1}\sum\limits_{j=0, j \neq i}^{m-1} (\alpha^{q^i} \beta^{q^j} - \beta^{q^i}\alpha^{q^j}) \delta^{q^j} T^{q^i+q^j}\right)\\
			& =& \left(\sum\limits_{i=0}^{m-2}\sum\limits_{j > i}^{m-1} (\alpha^{q^i} \beta^{q^j} - \beta^{q^i}\alpha^{q^j})( \delta^{q^j}-\delta^{q^i}) T^{q^i+q^j}\right)
		\end{eqnarray*}
		where we used $\alpha^{q^i} \beta^{q^j} - \beta^{q^i}\alpha^{q^j}=0$ for $i=j$. Clearly, the degree of this polynomial is at least $q+1$ and at most $q^{m-1}+q^{m-2}$.
	\end{proof}
	The next corollary is an immediate consequence of Lemma \ref{lemma:deg}. 
	\begin{corollary}
		Suppose that the function $f_{\alpha, \beta}(T, \delta T)$ is not identically zero over $\Fqm^*$. Then $f_{\alpha, \beta}(T, \delta T)$ can have at most $q^{m-1}+q^{m-2}- q-1$ many zeros in $\Fqms$.
	\end{corollary}
	Note that the polynomial $f_{\alpha, \beta}(T, \delta T)$ is divisible by $T^{q+1}$. Next, we determine the dimension of the code $\Cd$. To do so, we first determine a spanning set for the vector space $L_\delta(\Fq)$. We determine it in the next two lemmas.
	

	\begin{lemma}\label{lemma:monomialspan}
		Let $\delta$ be a fixed nonzero element of $\Fqm$. Let $d>1,\;d\mid m$ be the smallest positive integer such that $\delta$ is contained in the field $\F_{q^d}$ Then 
		\[
		L_\delta(\Fq)\subseteq\langle\{(\delta^{q^j}-\delta^{q^i}) T^{q^i+q^j}, 0 \leq i < j \leq m-1, d \nmid j-i  \}\rangle_{\Fqm}
		\]

	\end{lemma}
	\begin{proof}
		It is enough to prove that for each $\a,\beta\in\Fqm$, the determinantal function  $f_{\alpha, \beta}(T, \delta T)$ can be written as a $\Fqm$-linear combination of monomials in the set $\{(\delta^{q^j}-\delta^{q^i}) T^{q^i+q^j}, 0 \leq i < j \leq m-1, d \nmid j-i  \}$.
		Lemma \ref{lemma:deg} states that $$f_{\alpha, \beta}(T, \delta T)  = \left(\sum\limits_{i=0}^{m-2}\sum\limits_{j > i}^{m-1} (\alpha^{q^i} \beta^{q^j} - \beta^{q^i}\alpha^{q^j})( \delta^{q^j}-\delta^{q^i}) T^{q^i+q^j}\right) .$$ 
		
		The condition on $\delta$ implies that $\delta^{q^d}=\delta$. For any term of the form $X^{q^i+q^j}$ where $j = nd+i$, we obtain that the 
		$$
		\delta^{q^j}-\delta^{q^i} = \delta^{q^{nd+i}}-\delta^{q^i} = (\delta^{q^{nd}})^{q^i}  - \delta^{q^i} = \delta^{q^i}  - \delta^{q^i} =0.
		$$
		Therefore, the expansion of $f_{\alpha, \beta}(T, \delta T)$ has no terms of the form $T^{q^i+q^j}$ where $d \mid j-i$.

	\end{proof}
	In the next lemma we prove that both the spaces discussed in the last lemma are the same.
	
	\begin{lemma}
		\label{lemma:dim}
		Let $\delta$ be a fixed nonzero element of $\Fqm$ and let $d>1,\;d\mid m$ be the smallest positive integer such that $\delta$ is contained in the field $\F_{q^d}$. Then 
		\[
		L_\delta(\Fq)=\langle\{(\delta^{q^j}-\delta^{q^i}) T^{q^i+q^j}, 0 \leq i < j \leq m-1, d \nmid j-i  \}\rangle_{\Fqm}
		\]
		
	\end{lemma}
	
	\begin{proof}

		
		In the view of Lemma \ref{lemma:monomialspan}, we only have to show that for every 
		$0\le s< r\le m-1$, satisfying $d \nmid r-s$, there exist $\a,\;\beta\in\Fqm $ such that $(\delta^{q^r}-\delta^{q^s})T^{q^s+q^r}$ can be written as an $\Fq$-linear combination of some $f_{\alpha, \beta}(T, \delta T)$ for some $\a, \beta\in \Fqm$. Let $\gamma, \gamma^q, \ldots, \gamma^{q^{m-1}}$ be a normal basis for $\Fqm$ over $\Fq$. This implies that the matrix given by $\left( \gamma^{q^{i+j}} \right)_{0 \leq i,j \leq m-1}$ is nonsingular. Thus for any $0 \leq s \leq m-1$ there exists $b_{s,1}, b_{s,2}, \ldots, b_{s,m-1} \in \Fq$ such that
		$$ \sum\limits_{i = 0}^{m-1} b_{s,i} (\gamma^{q^{i}}, \gamma^{q^{i+1}}, \ldots, \gamma^{q^{i+m-1}} ) = e_{s}  $$ 
		where $e_s$ is the standard basis vector with a $1$ in $s^{\it th}$ position and zeroes everywhere else. As the vector $(\gamma^{q^{i}}, \gamma^{q^{i+1}}, \ldots, \gamma^{q^{i-1}} )$ is the coefficient vector for  $\Trqm(\gamma^{q^i} T)$, (omitting monomials not of the form $T^{q^j}$), taking the dot product of both sides with vector $(T, T^{q},\ldots, T^{q^{m-1}})$, we get 
		\begin{equation}
		\label{eq:normalbasis}
		\sum\limits_{i = 0}^{m-1} b_{s,i} \Trqm(\gamma^{q^i} T) = T^{q^s}.
		\end{equation}
		Thus, for each fixed $\beta\in\Fqm$, we have
		
		\begin{align*}
		&\sum\limits_{i = 0}^{m-1} b_{s,i} f_{\gamma^{q^i} ,\beta}(T, \delta \beta T) \\
		&=\sum\limits_{i = 0}^{m-1} b_{s,i} ( \Trqm(\gamma^{q^i}T)\Trqm(\delta \beta T) - \Trqm(\gamma^{q^i}\delta T)\Trqm(\beta T)) \\
		&= \sum\limits_{i = 0}^{m-1} b_{s,i} \Trqm(\gamma^{q^i}T)\Trqm(\delta \beta T) - \sum\limits_{i = 0}^{m-1} b_{s,i}\Trqm(\gamma^{q^i}\delta T)\Trqm( \beta T) \\
		&=  T^{q^s}\Trqm(\delta \beta T) - (\delta T)^{q^s} \Trqm( \beta T).
		\end{align*}
		
		\noindent	Now, taking $b_{r,1}, b_{r,2}, \ldots, b_{r,m-1}$ as in equation \eqref{eq:normalbasis} and consider the linear combination
		\begin{align*}
		&\sum\limits_{i=0}^{m-1}b_{r,i}(T^{q^s}\Trqm(\delta \gamma^{q^i} T) - (\delta T)^{q^s} \Trqm( \gamma^{q^i} T)) \\
		&=\sum\limits_{i=0}^{m-1}b_{r,i}T^{q^s}\Trqm(\delta \gamma^{q^i} T) - \sum\limits_{i=0}^{m-1}b_{r,i}(\delta T)^{q^s} \Trqm( \gamma^{q^i} T) \\
		&= T^{q^s}\sum\limits_{i=0}^{m-1}b_{r,i}\Trqm(\delta \gamma^{q^i} T) - (\delta T)^{q^s}\sum\limits_{i=0}^{m-1}b_{r,i} \Trqm( \gamma^{q^i} T) \\
		&= T^{q^s}(\delta T)^{q^r} - (\delta T)^{q^s}T^{q^r}\\
		&= (\delta^{q^r}-\delta^{q^s})T^{q^s+q^r}.
		\end{align*}
		In other words, for any $0\le r<s\le m-1$ with $d \nmid r-s$, 	monomials $(\delta^{q^r}-\delta^{q^s})T^{q^s+q^r}$ can be written as a linear combination of $ f_{\alpha, \beta}(T, \delta T)$ for some $\a,\; \beta\in\Fqm$. This completes the proof of the Lemma.
		%
		
	\end{proof}
	We have now found a simple basis for the space $L_\delta(\Fq)$. From this basis, it is clear that the dimension of the space $L_{\delta}(\Fq)$ is $\binom{m}{2}-\binom{\frac{m}{d}}{2}d$. Using this fact, we get the following corollary.

	\begin{corollary}
		\label{cor:projdim}
		Let $\delta\in\Fqm \setminus \Fq$ be a nonzero element such that the orbit $\Od$ is of size $N=\frac{q^m-1}{q-1}$. Let $1< d\le m$ be the smallest positive integer such that $d\mid m$ and $\delta\in\mathbb{F}_{q^d}$.  Then the code $\Cd$ has dimension $\binom{m}{2}-\binom{\frac{m}{d}}{2}d$. In particular, if $m$ is prime, then $\dim \Cd = \dim C(2, m).$
	\end{corollary}
	\begin{proof}
		
		As we have discussed, the dimension of the space $L_{\delta}(\Fq)= \binom{m}{2}-\binom{\frac{m}{d}}{2}d$. It is enough to show that the evaluation map defined in equation \eqref{eq: CodeCd} is an injective map. Note that the length $N$, of the code $\Cd$, is strictly bigger than $q^{m-1} + q^{ m-2}$ and a polynomial in $L_{\delta}(\Fq)$ can have degree at most  $q^{m-1} + q^{ m-2}$. In otherwords, no nozero polynomial from $L_{\delta}(\Fq)$ can map to zero, and hence the evaluation map defined in equation \eqref{eq: CodeCd} is injective.

	\end{proof}

	\section{Decoding Grassmann code $C(2, m)$}

	In this section, we propose a decoding algorithm for the Grassmann code $C(2, m)$ correcting up to $\dfrac{d-1}{2}$ errors. In \cite{BS}, a decoding algorithm for Grassmann codes $\Clm$ was proposed but unfortunately, the proposed algorithm can correct, asymptotically,  only up to $\lfloor(d-1)/2^{\ell+1}\rfloor$ errors. In other words, for $C(2, m)$ the proposed algorithm can correct around $\lfloor(d-1)/8\rfloor$ errors, which is far from the Grassmann code's error correcting capacity. Our decoding algorithm combines Reed--Solomon code decoding with information set decoding and the orbit structure of $C(2,m)$ to decode up to half the minimum distance. We begin this section defining a list decoder.

\begin{definition}
Let $C$ be an $[n,k,d]$ code over $\Fqm$. The code $C$ can be \emph{list decoded} correcting $\tau$ errors with list size $\eta$ if for any $y \in \Fqm^n$ and any $c \in C$ with $d_H(y,c) \leq \tau$ we can find a list $L \subseteq C$ of size at most $\eta$ containing $c$.
\end{definition}

Our decoding strategy is  to project a received word $r \in \Fq^{n}$, where $n=|\Gm|$, onto the different orbits $\Od$, decode the resulting projections and recover the original codeword from the positions in the projection. First,  we use Peterson's decoding algorithm to obtain a list of at most $q^m$ possible codewords on each of the orbits $\Od$. 

	\begin{theorem}[Petersen's decoding algorithm]
		An $[n,k,d=n-k+1]$ Reed--Solomon code can decode $\lfloor \frac{n-k+1}{2}\rfloor$ errors with complexity $\mathcal{O}(n^3)$.	
	\end{theorem}

	Now, we will give a decoding algorithm for the code $\Cd$. But before that we give the following remark:
	\begin{remark}
		\label{rmk:power} As earlier, let $\Od$ be an orbit in $\Gm$ with cardinality $|\Od|=N$, where $N=\frac{q^m-1}{q-1}$. We have seen that we may think the points of $\Od$ as $\gamma^i$ for some $i$, as $\gamma^i$ represents point $\langle \gamma^i, \gamma^i\delta\rangle_{\Fq}$ of $\Od$. Under this identification, $\Od=\{1, \gamma, \ldots, \gamma^{N-1}\}$. Without loss of generality we may assume the coordinates of $\Cd$ are indexed on the set $\Od$ in this fixed order. 
	\end{remark}

	\begin{lemma}
		\label{lem:orbitdecoder}
		We can decode up to $t=\frac{N-(q^{m-1}+q^{m-3}-q-1)}{2}$ errors for the code  $\Cd$ with list size $q^m$ and  complexity $\mathcal{O}(q^m\frac{(q^m-1)^3}{(q-1)^3})$.
		
	\end{lemma}
	\begin{proof}

		The proof is a little technical. Let $c = (c_0, c_1, \ldots, c_{N-1})\in {\Cd}$ be a transmitted codeword and let	$r = (r_0, r_1, \ldots, r_{N-1})$ be the received codeword with error vector $e = r-c$ with $\wt(e)\le t$. By decoding $c$ from $r$, we mean to find a polynomial function $f(T)\in L_{\delta}(\Fq)$ such that $ev_{\Od}(f)  = c$. Recall  from  Lemma \ref{lemma:deg} and Remark \ref{rmk:coeeficients},  that $f(T)$ can be written as 
		$$
		f(T)=a_1T^{q+1}+\dots +a_{k-1}T^{q^{m-1}+q^{m-3}} +a_{k}T^{q^{m-1}+q^{m-2}}
		$$
		where $k=\dim L_\delta(\Fq)$ and $a_i\in\Fqm$.  Instead of decoding $c$ from $r$, we shall decode $\widehat{c} = (c_0, \frac{c_1}{\gamma^{q+1}}, \ldots , \frac{c_{N-1}}{\gamma^{(q+1)(N-1)}})$  from $\widehat{r} = (r_0, \frac{r_1}{\gamma^{q+1}}, \ldots , \frac{r_{N-1}}{\gamma^{(q+1)(N-1)}})$, i.e., we try to find the polynomial $g(T)=f(T)/ T^{q+1}$ such that $ev_{\Od}(g)  =\widehat{c}$. Note that $\wt(\widehat{r}-\widehat{c})= \wt(r-c)\le t $. Also, $ev_{\Od}(g)$ is a codeword of a $RS_{q^m}(\Od, q^{m-1}+q^{m-2}-q)$ Reed--Solomon code. This Reed--Solomon code can not decode $t$ errors. Thus we must adapt Peterson's decoder. 

The polynomials $f$ and $g$ are sparse. The second highest term in $g(T)$ has degree at most $q^{m-1}+q^{m-3}-q-1$. That is: $$deg g(T) - a_{k}T^{q^{m-1}+q^{m-2}-q-1}  \leq q^{m-1}+q^{m-3}-q-1.$$ Instead of decoding $\widehat{r}$ as a codeword from $RS_{q^m}(\Od, q^{m-1}+q^{m-2}-q)$ we shall decode all $q^m$ possibilities $\widehat{r} - ev_{\Od}(b T^{q^{m-1}+q^{m-2}-q-1})$ for every $b\in \Fqm$ as codewords from $RS_{q^m}(\Od, q^{m-1}+q^{m-3}-q)$. Decoding is not guaranteed to work for most of the values of $b$. However, the codeword $\widehat{r} - ev_{\Od}(a_{k} T^{q^{m-1}+q^{m-2}-q-1})$ is contained in the smaller Reed--Solomon code and we can recover $ \widehat{g(T)} =a_1+\dots +a_{k-1}T^{q^{m-1}+q^{m-3}-q-1}$ from $\widehat{r} - ev_{\Od}(a_{k} T^{q^{m-1}+q^{m-2}-q-1})$ because few errors ocurred. 

Our list decoder works as follows.  For each $b \in \F_{q^m}$  attempt to decode $\widehat{r} - ev_{\Od}(bT^{q^{m-1}+q^{m-2}-q-1})$ as a codeword from a $RS_{q^m}(\Od, q^{m-1}+q^{m-3}-q)$ Reed--Solomon code. If this decoding attempt is successful, then we find some $ \widehat{g_b(T)} =b_1+\dots +b_{k-1}T^{q^{m-1}+q^{m-3}-q-1}$. In this case we add $ev_{\Od}(g_b(T) + bT^{q^{m-1}+q^{m-2}-q-1})$ to the list of possible codewords. Because less than $\frac{N-(q^{m-1}+q^{m-3}-q-1)}{2}$ errors occurred the codeword  $\widehat{c} - ev_{\Od} (a_k T^{q^{m-1}+q^{m-2}-q-1})$ is obtained when decoding with $b = a_k$. Therefore we add $ev_{\Od}(g(T))$ to the list of possible codewords. 		\end{proof}

	Our decoding algorithm uses a combination of Reed--Solomon decoding and Information set decoding. 

\begin{definition}[Information set decoding]

Let $C$ be an $[n,k,d]$ code. Let $\mathcal{I}$ be a collection of information sets of $C$. We can decode up to $t$ errors with $\mathcal{I}$ if for any set $T$ of $t$ positions there exists an information set $I \in \mathcal{I}$ such that $T \cap I = \emptyset$.

\end{definition}
The way information set decoding works is by taking a received word $r$ and all the information sets in $\mathcal{I}$. For each $I \in \mathcal{I}$ one encodes the projection $r^I$ as a codeword of $C$. If $t$ errors or less ocurred, there exists an information set with no errors. The codeword corresponding to that set of positions will be at distance $\leq t$ from our received word. In the next lemma we give a bound for number of elements in a field extension not lying in any proper subfield of the extension field. This bound is needed to count the orbits with contain an information set for the Grassmann code $C(2,m)$. These orbits are crucial for our decoder.

	\begin{lemma}
		\label{lemma: bound} 
		Let $\F_{q^m}$ be a finite field with $q^m$ elements. If $m \geq 3$, then there are at least $q^{m}- q^{m-2}$ elements in $\Fqm$ not lying in any subfield of $\Fqm$ containing $\Fq$. 
		
	\end{lemma}
	\begin{proof}
		We first assume that $m>4$. Let $p_1,\ldots, p_r$ be all distinct primes dividing $m$. 	Then for each $1\le i\le r$, there is a unique field of degree $q^{\frac{m}{p_i}}$ and this is a maximal subfield of $\Fqm$. Therefore there are at least $q^m-\sum_{i}q^{\frac{m}{p_i}}$ elements not lying in any subfield of $\Fqm$. Let $p=p_1$ be the smallest prime dividing $m$. Then there are at least 
		\[
		q^m -\sum q^{\frac{m}{p_i}} \ge q^m -q^{\frac{m}{p}+1}
		\]
		elements is the field $\Fqm$ that does not lie in any proper subfield of $\Fqm$. Since $m\ge 5$ and $p\ge 2$, we have $q^m -q^{\frac{m}{p}+1}\ge q^m -q^{m-2}$ . Clearly,  $q^m -q^{m-2}\ge q^{m-1}+ q^{m-2}$ which proves that there are at least $q^{m-1}+ q^{m-2}$ elements of the field $\Fqm$ that does not lie in any proper subfields of $\Fqm$.  On the other hand if $m=4$, then $2$ is the only prime dividing $4$ and there is a unique subfield of degree $2$ in $\mathbb{F}_{q^4}$, namely $\mathbb{F}_{q^2}$. Therefore, there are $q^4-q^2$ elements in the field $\mathbb{F}_{q^4}$ that do not lie in the subfield $\mathbb{F}_{q^2}$. But $q^4-q^2\ge q^3+q^2$.

For the case $m = 3$, as $m$ is prime, the only subfield of $\mathbb{F}_{q^3}$ containing $\mathbb{F}_q$ is $\mathbb{F}_q$  itself. Therefore there are $q^3-q = q^{m}-q^{m-2}$ elements in $\Fqm$ not lying in any subfield of $\Fqm$ containing $\Fq$. 

 This completes the proof of the lemma.
	\end{proof}	
	%
	%

In this lemma we use the Pigeonhole principle to prove that if few errors ocurr, then there exists an orbit of $C(2,m)$ with an information set containing few errors.
	
	\begin{lemma}\label{lem:errorsetorbits}
Let $E \subseteq \G$. If $|E| \leq \lfloor(d-1)/2\rfloor$ then there exists an orbit $\Od$ with an information set such that $|\Od \cap E| < \frac{N-(q^{m-1}+q^{m-3}-q-1)}{2}$ 
	\end{lemma}
	\begin{proof}
			
		From From Corollary \ref{cor:projdim}, we know that if $\delta \in \Fqms$ but $\delta$ does not lie in any proper subfield of $\Fqm$, then the corresponding code $\Cd$ is of dimension ${m\choose 2}$. From Lemma \ref{lemma: bound}, we know that there are at least $q^{m}- q^{m-2}$ many elements in $\Fqms$ that does not lie in any proper subfield of $\Fqm$. These elements will correspond to $\frac{q^{m} - q^{m-2}}{q^2-q} = q^{m-2}+q^{m-3}$ orbits. In other words, there are at least $q^{m-2}+ q^{m-3}$ orbits $\Od$ with an information set of $C(2, m)$.

As $|E| \leq \lfloor\frac{d-1}{2} \rfloor$, where $d = q^{2(m-2)}$ note that
		\[ \frac{|E|}{\frac{N-(q^{m-1}+q^{m-3}-q-1)}{2}} \leq
		\lfloor\frac{\frac{q^{2(m-2)}-1}{2}}{\frac{N-(q^{m-1}+q^{m-3}-q-1)}{2}}\rfloor < q^{m-2} + q^{m-3}
		\]
		there exists one orbit $\Od$	 that contains no more than $\frac{N-(q^{m-1}+q^{m-3}-q-1)}{2}$ elements in $E$. \end{proof}

	Now we are ready to prove the main result of this article. In the next theorem, we give a decoding algorithm for the Grassmann code $C(2, m)$ which decodes up to $\lfloor (d-1)/2 \rfloor$ errors. Our algorithm partitions the codeword of $C(2,m)$ onto the different orbits $\Od$, then decodes each the projection on each orbit using the list decoder of Lemma \ref{lem:orbitdecoder}. Note that for each orbit we obtain at most $q^m$ possible projections on each orbit. We then try to recover the original codeword from each of the $q^m$ possibilities on each of the orbits. If less than $\lfloor(d-1)/2\rfloor$ errors occur, then we will find our original codeword on this list and it will be the closest codeword to the received word.

\begin{algorithm}[Orbit Projection decoder for $C(2,m)$].
  \\ 
\begin{itemize}
\item Input: $r  = c + e$ where $c \in C(2,m)$.
\item Initially let $L$ be an empty list.
\item For each orbit $\Od$ check if the orbit contains an information set.
\item If so, project $r$ onto the orbit and apply the list decoder of Lemma \ref{lem:orbitdecoder}.
\item Use an information set in $\Od$ to encode a codeword from $C(2,m)$.
\item Add this codeword to $L$.
\item Return $L$.
\end{itemize}

\end{algorithm}

	\begin{theorem}
		Let $m\ge 4$ and let  $C(2, m)$ be the corresponding Grassmann code. Using the projection of $C(2, m)$ onto some orbit $\Od$  with $|\Od|=N$ then the algoritm described corrects up to $\lfloor(d-1)/2\rfloor$ errors for Grassmann code $C(2, m)$ where $d= q^{2(m-2)}$.
	\end{theorem}
	
	\begin{proof}  Input: $r  = c + e$ where $c \in C(2,m)$ where $wt(e) < \lfloor(d-1)/2\rfloor$. As less then $\lfloor(d-1)/2\rfloor$ errors ocurred, Lemma \ref{lem:errorsetorbits} implies there exists an orbit with an information set and less than $\frac{N-(q^{m-1}+q^{m-3}-q-1)}{2}$ errors. When the algorithm decodes $r$ restricted to this orbit,  Lemma \ref{lem:orbitdecoder} implies that the restriction of $c \in C(2,m)$ onto the orbit is in the returned list. The decoding algorithm then uses the information set in $\Od$ to recover the codeword $c \in C(2,m)$ from its projection onto $\Od$. Then the algorithm adds $c$ to the output list. Therefore, if less than$\lfloor(d-1)/2\rfloor$ errros occured, the correct codeword is on the list. \end{proof}

	\section{Acknowledgements}
	
	For the duration of this work, the first named author  was supported by the National Institute of General Medical Sciences of the National Institutes of Health, The United States of America under Award Number R25GM121270 . The content is solely the responsibility of the authors and does not necessarily represent the official views of the National Institutes of Health .The second named author would like to express his gratitude to the Indo-Norwegian project supported by Research Council of Norway (Project number 280731), and the DST of Govt. of India. 

\end{document}